\documentclass[11pt]{article}

\usepackage[top=25mm, bottom=25mm, left=31mm, right=31mm]{geometry}
\usepackage{amsmath,amsthm,amssymb}
\usepackage{cases}
\usepackage[document]{ragged2e}

\newenvironment{theorem}[2][Theorem]{\begin{trivlist}
\item[\hskip \labelsep {\bfseries #1}\hskip \labelsep {\bfseries #2.}]}{\end{trivlist}}
\newenvironment{lemma}[2][Lemma]{\begin{trivlist}
\item[\hskip \labelsep {\bfseries #1}\hskip \labelsep {\bfseries #2.}]}{\end{trivlist}}

\newenvironment{remark}[2][Remark]{\begin{trivlist}
\item[\hskip \labelsep {\bfseries #1}\hskip \labelsep {\bfseries #2.}]}{\end{trivlist}}
\newenvironment{acknowledgement}[2][Acknowledgement]{\begin{trivlist}
\item[\hskip \labelsep {\bfseries #1}\hskip \labelsep {\bfseries #2.}]}{\end{trivlist}}
\newenvironment{$proof$}[2][$Proof$]{\begin{trivlist}
\item[\hskip \labelsep {\bfseries #1}\hskip \labelsep {\bfseries #2.}]}{\end{trivlist}}

\newenvironment{corollary}[2][Corollary]{\begin{trivlist}
\item[\hskip \labelsep {\bfseries #1}\hskip \labelsep {\bfseries #2.}]}{\end{trivlist}}
\begin{document}


\title{On states of quantum theory}

\author{Amir R. Arab$^{1, 2}$\\$^{1}$Steklov Mathematical Institute of Russian Academy of
Sciences,\\ Gubkina str., 8, Moscow 119991, Russian Federation;\\
$^{2}$Moscow Institute of Physics and Technology, 9 Institutskiy per.,\\
Dolgoprudny, Moscow Region, 141701, Russian Federation\\
\texttt{amir.arab@phystech.edu}\\}

\maketitle

\begin{abstract}
\begin{justify}
In this paper the generalized quantum states, i.e. positive and normalized linear functionals on $C^{*}$-algebras, are studied. Firstly, we study normal states, i.e. states which are represented by density operators, and singular states, i.e. states can not be represented by density operators. It is given an approach to the resolution of bounded linear functionals into quantum states by applying the GNS construction, i.e. the fundamental result of Gelfand, Neumark and Segal on the representation theory of $C^{*}$-algebras, and theory of projections. Secondly, it is given an application in quantum information theory. We study covariant cloners, i.e. quantum channels in the Heisenberg and the Schr\"{o}dinger pictures which are covariant by shifting, and it is shown that the optimal cloners can not have a singular component. Finally, we discuss on the representation of pure states in the sense of the Gelfand-Pettis integral. We also give physical interpretations and examples in different sections of the present work.
\end{justify}
\end{abstract}
\vspace{5mm} 
\textbf{Keywords:}\, States on $C^{*}-$algebras; GNS construction; Covariant cloning; Gelfand-Pettis integral
\section{Introduction}
\begin{justify}
A generalized quantum state is a state on a $C^{*}-$algebra that generalizes the notion of a density matrix in quantum mechanics. A density matrix represents a quantum state, both a mixed state and a pure state. A density matrix in turn generalizes the notion of a state vector, which only represents a pure state. In the $C^{*}-$algebraic formulation of quantum mechanics, states in the previous sense correspond to physical states, i.e. mappings from physical observables to their expected measurement outcome.\par
In 1943, Gelfand and Neumark deﬁned what is now called a $C^{*}-$algebra and proved the basic theorem that every $C^{*}-$algebra is isomorphic to the norm-closed $*$-algebra of operators on a Hilbert space. Their paper [1] also contained the rudiments of what is now called the GNS construction, connecting states to representations. In its present form, this construction is due to Segal [2], a great admirer of von Neumann, who generalized von Neumann’s idea of a state as a positive normalized linear functional from $B(H)$ to arbitrary $C^{*}-$algebras. Moreover, Segal returned to von Neumann’s motivation of relating operator algebras to quantum mechanics.\par
Starting with Segal in 1947, $C^{*}-$algebras have become an important tool in mathematical physics, where traditionally most applications have been to quantum systems with inﬁnitely many degrees of freedom, such as quantum statistical mechanics in inﬁnite volume and quantum ﬁeld theory. Many important examples of states in quantum physics (for instance, the Gibbs states for free Bose gas) are normal states, i.e. they are of the form $\rho(a)= \mathrm{tr}(\Lambda a)$ with a certain density matrix $\Lambda$.\par
In 1966, Dixmier proved that there exist states on $B(H)$, the bounded linear operators on a separable complex Hilbert space, which are not normal [3]. Dixmier states have the further property to be ``singular'', i.e. they vanish on the finite rank operators. Quantum states have been classified in a wide range of studies, for instance see [4-8]. We are motivated by two papers [9, 10] which studied states on $B(H)$ in the classical case by considering finitely additive measures and ultrafilters, and provided some results on representation of them. In the present paper, states on $B(H)$ are generalized to states on a $C^{*}-$algebra.

This paper is devoted to the description of generalization of quantum states, applying this treatment in quantum mechanics, quantum field theory, quantum dynamics and quantum information theory. This paper is organized as follows. In Section 2, we provide basic definitions and theorems on weak topologies on $B(H)$, $C^{*}-$algebras, von Neumann algebras, and representation of $C^{*}-$algebras. In Section 3, we introduce different types of states and formulate a fundamental theorem on normal states. We also address the formulation of quantum mechanics in terms of worlds that is verified by realization of singular states. In Section 4, we introduce the GNS construction and universal representation of a $C^{*}-$algebra. We also provide an approach to quantum field theory. In Section 5, we prove some results on extension of bounded linear functionals, topological properties with respect to weak topologies on $B(H)$, relation between a representation of a $C^{*}-$algebra and its universal representation, and classification of normal and singular states by a central projection. We also deal with generalized states in quantum dynamics. In Section 6, we introduce phase space, Weyl operators and covariant cloning, and by applying the notion of fidelity it is shown that the optimal cloner can not have a singular component. And finally in Section 7, we have a short discussion on the representation of pure states in the sense of the Gelfand-Pettis integral. We also discuss on the process of measurement in the formalism of quantum mechanics in terms of worlds.
\end{justify}
\section{Preliminaries}
\begin{justify}
Let $B(H)$ be the set of bounded linear operators on a complex Hilbert space $H$. The induced topology on $B(H)$ by the family of semi-norms,\\
$(i)\; p_{x}(T)=\|T(x)\|$ is called the \emph{strong-operator topology }(SOT); \\
$(ii)\; p_{x, y}(T)=|\langle T(x), y\rangle|$, is called the \emph{weak-operator topology }(WOT); \\
$(iii)\; p_{\{x_{n}, y_{n}\}}(T)=|\sum_{n}\langle T(x_{n}), y_{n}\rangle|$, where $\sum_{n}(\|x_{n}\|^{2}+\|y_{n}\|^{2})<\infty$ is called the \emph{ultraweak topology}.\\
The SOT and the ultraweak topology are stronger than the WOT. On every norm-bounded set the WOT and the ultraweak topology are the same, particularly the unit ball is compact in both topologies.\par
We assume that $\textsf{A}$ is an associative algebra over the field $\mathbb{C}$ with the unit element $1_{\textsf{A}}$. Let $\|.\|$ be a norm on $\textsf{A}$, as a vector space on $\mathbb{C}$, and $*:\,$\textsf{A}$\longrightarrow $\textsf{A}$, a\longmapsto a^{*} $, be an antilinear map. Then $(\textsf{A},\|.\|,* )$ is called a \emph{$C^{*}-$algebra}, if $(\textsf{A},\|.\|)$ is complete and for each $a, b$ in $\textsf{A}$, we have:\\
$1.\quad a^{**}=a$\qquad (\emph{involution}),\\
$2.\quad (ab)^{*}=b^{*}a^{*}$\qquad (\emph{anti-homomorphism}),\\
$3.\quad \|ab\|\leq \|a\|\|b\|$\qquad (\emph{submultiplicativity}),\\
$4.\quad \|a^{*}\|=\|a\|$\qquad (\emph{isometry}),\\
$5.\quad \|a^{*}a\|=\|a\|^{2}$\qquad (\emph{$C^{*}-$property}).\par
$B(H)$ with the usual operator norm is a $C^{*}-$algebra. An algebra $\textsf{A}$ is called \emph{self-adjoint} if for each $a\in \textsf{A}$, $a^{*}\in \textsf{A}$ is so. A \emph{von Neumann algebra} is a self-adjoint algebra $\textsf{A}\subseteq B(H)$ that is closed with respect to the WOT. Every von Neumann algebra is a $C^{*}-$ algebra, but the converse is not true, for instance the $C^{*}-$ algebra $K(H)$ of compact operators on an infinite dimensional Hilbert space $H$ is not a von Neumann algebra. Let $\textsf{A}^{-}$ be the strong- (equivalently weak-)operator closure of a $C^{*}-$algebra $\textsf{A}\subseteq B(H)$. According to the Kaplansky density theorem, $(\textsf{A})_{1}$, the unit ball of $\textsf{A}$, is strong-operator dense in $(\textsf{A}^{-})_{1}$, the unit ball of $\textsf{A}^{-}$. The \emph{predual} of a von Neumann algebra $\textsf{A}$ is the linear space of linear functionals on $\textsf{A}$ that are weak-operator continuous on $($\textsf{A}$)_{1}$ and we denote it by $(\textsf{A})_{*}$.\par
A \emph{representation} of a $C^{*}-$algebra $\textsf{A}$ on $H$ is a $*-$homomorphism $\pi:\textsf{A}\longrightarrow B(H)$ (i.e. homomorphism with this property $\pi(a^{*})=\pi(a)^{*}$ for each $a$ in $\textsf{A}$). If, further, $\pi$ is one-to-one, it is called $*-isomorphism$ or \emph{faithful} representation. Every $*-$homomorphism is norm decreasing and then is continuous. The image of a $C^{*}-$algebra under a $*-$homomorphism is a $C^{*}-$algebra, specially $\pi(\textsf{A})$ is a $C^{*}-$subalgebra of $B(H)$. If there is a vector $x$ in $H$ for which the linear subspace
\begin{align*}
\pi(\textsf{A})x=\{\pi(a)x : a\in \textsf{A}\},
\end{align*}
is dense in $H$, $\pi$ is called a \emph{cyclic} representation, and $x$ is termed as a \emph{cyclic vector} for $\pi$. Two representations $\pi_{1}:\textsf{A}\longrightarrow B(H_{1})$ and $\pi_{2}:\textsf{A}\longrightarrow B(H_{2})$ are \emph{equivalent}, if there is a unitary operator $U:H_{1}\longrightarrow H_{2}$ such that $U\pi_{1}(a)=\pi_{2}(a)U$ where $a\in \textsf{A}$.
\section{States}
A \emph{state} on a $C^{*}-$algebra $\textsf{A}$ is a linear functional $\rho: \textsf{A} \longrightarrow \mathbb{C}$ such that $\rho(a^{*}a)\geq 0$ for each $a$ in $\textsf{A}$ (\emph{positivity}) and $\rho(1_{\textsf{A}})=1$ (\emph{normalization}). For example, when $\textsf{A}\subseteq B(H)$ then every unit vector $x\in H$ defines a state $\omega_{x}(T):=\langle T(x), x\rangle$ that is called a \emph{vector state}. A state $\rho$ on $\textsf{A}$ is called \emph{faithful} if $\rho(a^{*}a)> 0$ unless $a=0$. The set of states on $\textsf{A}$ is called the \emph{state space} of $\textsf{A}$ and we denote it by $S(\textsf{A})$. The state space is a compact (with respect to the weak$-^{*}$ topology) convex set. For example the state space of the algebra of $2\times 2$ complex matrices is the Bloch sphere (the unit ball in $\mathbb{R}^{3}$). A state $\rho$ on $\textsf{A}$ is called \emph{pure (extreme)} if it is not a convex combination of other states, i.e. if $\rho=t\rho_{1}+(1-t)\rho_{2}$ for some $0< t< 1$, then $\rho=\rho_{1}=\rho_{2}$. In fact pure states are elements of the boundary of $S(\textsf{A})$. A state that is not pure is called \emph{mixed}. A state $\rho$ on a $C^{*}-$algebra $\textsf{A}$ is called \emph{normal}, if $a_{i}\nearrow a$ then $\rho(a_{i})\nearrow \rho (a)$, i.e. if for every increasing net $\{a_{i}\}_{i\in I}$ of positive elements of $\textsf{A}$ (i.e. elements may be written in the form $b^{*}b$), we have $\rho$(sup$_{i\in I}\,a_{i})$=sup$_{i\in I}\,\rho(a_{i})$. For example with $\{e_{i}\}_{i=1}^{n}$ an orthonormal basis of $\mathbb{C}^{n}$, $\rho(T)=\sum_{i=1}^{n}p_{i}\langle T(e_{i}), e_{i}\rangle$, where $p_{i}\geq 0$ ($1\leq i\leq n$) and $\sum_{i=1}^{n}p_{i}=1$, is a normal state on the algebra $\mathbb{C}^{n\times n}$. For a unit vector $e$, $\rho(T)=\langle T(e), e\rangle$ is a normal pure state.\par
\begin{theorem}{1} ([6])
\emph{The following conditions for a state $\rho$ on a von Neumann algebra $\textsf{A}$ are equivalent,}\\
\emph{1. $\rho$ is normal;}\\
\emph{2. $\rho(T)=\mathrm{tr}(\Lambda T)$ for some density operator $\Lambda$ on $H$ ($\Lambda\geq 0, \mathrm{tr}(\Lambda)=1$);}\\
\emph{3. $\rho$ is ultraweakly continuous;}\\
\emph{4. $\rho$ is weak-operator continuous on $(\textsf{A})_{1}$;}\\
\emph{5. $\rho$ is strong-operator continuous on $(\textsf{A})_{1}$;}\\
\emph{6. $\rho=\sum_{n=1}^{\infty}\omega_{x_{n}}$ where $\sum_{n=1}^{\infty}\|x_{n}\|^{2}=1$;}\\
\emph{7. $\rho(\vee_{i\in I}\,e_{i})=\sum_{i\in I}\,\rho(e_{i})$ for every orthogonal set of projections $\{e_{i}\}_{i\in I}$ $(\mathrm{i.e.}\; e_{i}^{2}=e_{i}=e_{i}^{*}\;  \mathrm{and}\;e_{i}e_{j}=\delta_{ij}e_{i})$.}
\end{theorem}\par
\subsection{Physical interpretation and example}
We do not measure vectors in a Hilbert space. What we measure are expectation values. This is precisely the meaning of a state $\rho$ as deﬁned above. A state assigns an expectation value to each physical observable. That is what a physical state is: a list of expectation values for all our observables satisfying reasonable postulates. It means that if $a$ is self-adjoint, then $\rho(a)$ is the expectation value of $a$ and $\rho(a^{2})-(\rho(a))^{2}$ is the variance. Therefore, the deﬁnition of state includes not only expectation values of observables, but also their moments.
\subsubsection{Formulation of quantum mechanics in terms of worlds}
Let $\mathcal{Q}$ be a quantum system with a separable Hilbert space $H$. A \emph{world} of $\mathcal{Q}$ is an orthonormal basis of $H$ [22]. Since there are many orthonormal bases for a Hilbert space $H$, then there are many worlds for $\mathcal{Q}$. For two worlds $W=\{e_{n}\}$ and $W^{\prime}=\{e^{\prime}_{n}\}$, $W$ is identified with $W^{\prime}$ whenever for each $n$, we have $e^{\prime}_{n}=c_{n}e_{n}$ where $c_{n}$ is a complex number with $|c_{n}|=1$. For a given world $W=\{e_{n}\}$ and unitary operator $U$, $W^{\prime}=\{Ue_{n}\}$ is also a world. Furthermore, any operator mapping an orthonormal basis into another one is a unitary operator. It results that the evolution of worlds is described by $U_{t+s}W=U_{t}U_{s}W$ ($t, s\geq 0$), where $W_{0}=W$ and $U_{t}$'s are unitary operators. Consequently, the evolution satisfies the following Schr\"{o}dinger equation:
\begin{gather*}
\frac{dW_{t}}{dt}=iBW_{t}\quad (t>0),
\end{gather*}
with $W_{0}=W$, where $B$ is a self-adjoint operator.

For a certain world $W=\{e_{n}\}$ of $\mathcal{Q}$, an \emph{observable} is a self-adjoint operator on $H$ which is diagonal in the basis $\{e_{n}\}$. The set of all observables relative to $W$, is denoted by $\mathcal{O}(W)$.

Let $\mathcal{Q}=\mathcal{Q}_{1}\times\mathcal{Q}_{2}$ be a bipartite system composed of two subsystems $\mathcal{Q}_{1}$ and $\mathcal{Q}_{2}$, with Hilbert space $H_{1}\otimes H_{2}$ where $H_{1}$ and $H_{2}$ are Hilbert spaces of $\mathcal{Q}_{1}$ and $\mathcal{Q}_{2}$, respectively. For two given observables $O_{1}$ and $O_{2}$ in the worlds $W_{1}=\{e_{j}^{1}\}$ (of $\mathcal{Q}_{1}$ ) and $W_{2}=\{e_{k}^{2}\}$ (of $\mathcal{Q}_{2}$ ), respectively, $O_{1}\otimes O_{2}$ is an observable in the world $W=\{e_{j}^{1}\otimes e_{k}^{2}\}$ of $\mathcal{Q}$. Conversely, for a given observable $O$ in $W=\{e_{j}^{1}\otimes e_{k}^{2}\}$, there are observables $O_{1}$ and $O_{2}$ in the worlds $W_{1}=\{e_{j}^{1}\}$ and $W_{2}=\{e_{k}^{2}\}$ such that $O=O_{1}\otimes O_{2}$. Similarly, we can study observables in a world of multipartite systems.

In the formalism of quantum mechanics in terms of worlds, observables are considered \emph{relative} to a certain world of the system. This relativity goes back to the fundamental works of Dirac and 't Hooft. Dirac in his principles of quantum mechanics [23] introduced the notion of orthonormal representation as an orthonormal basis and 't Hooft in his deterministic quantum mechanics [24] introduced the notion of ontological basis as an orthonormal basis.

Let $\mathcal{M}(W)$ be the von Neumann subalgebra of $B(H)$ generated by all bounded operators in $\mathcal{O}(W)$. A \emph{state} in $W$ is defined to be a state on $\mathcal{M}(W)$. For a world $W=\{e_{n}\}$, $w_{n}(T)=\langle T(e_{n}), e_{n}\rangle$ defines a pure state corresponding to the pure state $|e_{n}\rangle$ in the conventional formulation of quantum mechanics.

We will continue our discussion on this formalism and realization of singular states (which will be defined in Section 5) in Section 7.1.
\section{The GNS construction}
The basic idea of the GNS construction is that given an algebra $\textsf{A}$ of observables and a state $\rho$ on this algebra, we can construct the Hilbert space on which the algebra of observables acts on. From the mathematical point of view, the algebra of observables is a $C^{*}-$algebra.\par
The \emph{Gelfand-Neumark-Segal (GNS) theorem} constructs a representation of a $C^{*}-$algebra $\textsf{A}$ starting from a state $\rho:\textsf{A}\longrightarrow \mathbb{C}$. The main idea is to define an inner product via $\langle a, b\rangle:=\rho(b^{*}a)$. If $\rho$ is not faithful, one has to define the elements of the sought Hilbert space in terms of equivalent classes and consider the set $\textsf{A} / N$ with $N=\{a\in \textsf{A} : \rho(a^{*}a)=0\}$. By completing $\textsf{A} /N$ we get a Hilbert space $H_{\rho}$ and define $\pi_{\rho}(a)$ as the left multiplication by $a$ on $H_{\rho}$, i.e. $\pi_{\rho}(a)(b)=ab$. We can summarize this discussion in the following theorem,
\begin{theorem}{2}([1])
\emph{Given a state $\rho:\textsf{A}\longrightarrow \mathbb{C}$, there is a cyclic representation $\pi_{\rho}:\textsf{A}\longrightarrow B(H_{\rho})$ with a unit cyclic vector $x_{\rho}\in H_{\rho}$, such that $\rho=\omega_{x_{\rho}}\circ \pi_{\rho}$; that is $\rho(a)=\langle \pi_{\rho}(a)x_{\rho}, x_{\rho}\rangle$ where $a\in \textsf{A}$. If $\pi:\textsf{A}\longrightarrow B(K)$ is another cyclic representation with unit cyclic vector $x$ such that $\rho(a)=\langle \pi(a)x, x\rangle$, then $\pi_{\rho}$ and $\pi$ are equivalent.}
\end{theorem}\par
For example with $\textsf{A}=\mathbb{C}^{2\times 2}$, $H_{\rho}\cong \mathbb{C}^{4}\; \mathrm{or}\; \mathbb{C}^{2}$ [11]. One of the important results of the GNS theorem is that every $C^{*}-$algebra has a faithful representation (Gelfand-Neumark theorem). In other words, for every $C^{*}-$algebra $\textsf{A}$, there is a Hilbert space $H$ such that $\textsf{A}$ is a $*-$isomorphic to a $C^{*}-$subalgebra of $B(H)$. Suppose that $\textsf{A}$ is a $C^{*}-$algebra with the state space $S(\textsf{A})$. The representation $\Pi=\bigoplus_{\rho\in S(\textsf{A})}\pi_{\rho}$ of $\textsf{A}$ is faithful and is called the \emph{universal representation} of $\textsf{A}$.\par
\subsection{Physical interpretation and example}
The GNS construction is closely related to constructing a Hilbert space from a vacuum state by acting on it with creation operators and thereby adding particles. If we think $a$ and $a^{*}$ as annihilation and creation operators, respectively, the meaning of $\langle a, a\rangle:=\rho(a^{*}a)$ is the expectation value of the particle number for the state $\rho$. When a representation is irreducible, every non-zero vector is cyclic—by using annihilation operators, one can get to the ground state. One should take this comment with a grain of salt due to domain and distribution issues. A ﬁnite-dimensional and completely rigorous example occurs in the theory of spin by use of ladder operators [12].
\subsubsection{Approach to quantum field theory}
Given an $n$-dimensional manifold $\mathcal{M}$, the \emph{cotangent bundle} of $\mathcal{M}$ is defined to be the 2$n$-dimensional manifold comprised by all points $x$ of $\mathcal{M}$ together with all cotangent vectors at $x$. The \emph{phase space} $\mathcal{P}$ is defined to be the cotangent bundle of an $n$-dimensional configuration manifold $\mathcal{M}$ such that the symplectic form $\Omega_{ab}$, i.e. nondegenerate closed two-form, on $\mathcal{P}$ is given by
\begin{gather*}
\Omega_{ab}=\sum_{\mu=1}^{n}2(\nabla p_{\mu})_{[a}(\nabla q_{\mu})_{b]},
\end{gather*}
where $(q_{1},...,q_{n};p_{1},...,p_{n})$ are canonical coordinates on $\mathcal{P}$.

In classical mechanics, an \emph{observable} is a smooth map $f:\mathcal{P}\longrightarrow \mathbb{R}$. The set of classical observables $\mathcal{O}_{c}$ is an infinite dimensional vector space. The inverse symplectic form $\Omega^{ab}$ gives rise to an algebraic structure on $\mathcal{O}_{c}$, called the \emph{Poisson bracket}, which is defined by
\begin{gather*}
\{f, g\}=\Omega^{ab}\nabla_{a}f\nabla_{b}g.
\end{gather*}

In quantum mechanics, an \emph{observable} is a self-adjoint operator on an infinite-dimensional separable Hilbert space $\mathcal{F}$. The set of quantum observables is denoted by $\mathcal{O}_{q}$. Let $\hat{}:\mathcal{O}_{c}\longrightarrow \mathcal{O}_{q}$ be the map taking classical observables to quantum observables such that for any pair of classical observables $f, g$ we have
\begin{gather*}
[\hat{f}, \hat{g}]=i\{\widehat{f, g}\},
\end{gather*}
where $\hbar=1$.

Let $\mathcal{P}$ be a symplectic vector space, i.e. a vector space such that the symplectic form $\Omega_{ab}$ has constant components in a globally parallel basis. Therefore, $\Omega$ may be considered as an antisymmetric bilinear map $\Omega:\mathcal{P}\times \mathcal{P}\longrightarrow \mathbb{R}$ rather than a tensor field. In this notation, the Poisson bracket of the fundamental observables, i.e. linear functions, is given by
\begin{gather*}
\{\Omega(\psi_{1}, .), \Omega(\psi_{2}, .)\}=-\Omega(\psi_{1}, \psi_{2}),
\end{gather*}
and consequently
\begin{gather}
[\hat{\Omega}(\psi_{1}, .), \hat{\Omega}(\psi_{2}, .)]=-i\Omega(\psi_{1}, \psi_{2})I,
\end{gather}
for all $\psi_{1}, \psi_{2}\in \mathcal{P}$, where $I$ is the identity operator. It is convenient to work with the exponential version of (1). Let $W(\psi)=e^{i\Omega(\psi, .)}$ and seek a correspondence map "\, $\widehat{}$ \," for which $\hat{W}(\psi)$ is unitary, varies continuously with $\psi$ in the SOT and satisfies the following relations:
\begin{gather*}
\hat{W}(\psi_{1})\hat{W}(\psi_{2})=e^{\frac{i\Omega(\psi_{1}, \psi_{2})}{2}}\hat{W}(\psi_{1}+\psi_{2}),\quad \hat{W}^{*}(\psi)=\hat{W}(-\psi).
\end{gather*}
These relations known as the \emph{Weyl relations} which uniquely determine $(\mathcal{F}, \hat{W}(\psi))$.

In the framework of general relativity, spacetime structure is defined by a four-dimensional manifold $M$, on which there is a Lorentz metric $g_{ab}$. Furthermore, $g_{ab}$ is related to the matter distribution by Einstein's equation. For the spacetime structure $(M, g_{ab})$, where $g_{ab}$ is smooth, we study the formulation of the quantum theory of a Klein-Gordon scalar field $\phi$. The curved spacetime version of the Klein-Gordon equation is
\begin{gather}
\nabla^{a}\nabla_{a}\phi-m^{2}\phi=0,
\end{gather}
where $\nabla_{a}$ is the derivative operator compatible with $g_{ab}$. Let $\Sigma\subset M$ be any closed set which is achronal, i.e. no pair of points $P, Q\in\Sigma$ can be joined by a timelike curve. The \emph{domain of dependence} of $\Sigma$ is the set of all points $P$ of $M$ such that every (past and future) inextendible causal curve [25] through $P$ intersects $\Sigma$ and it is denoted by $D(\Sigma)$. If $D(\Sigma)=M$, then $\Sigma$ is said to be a \emph{Cauchy surface} for the spacetime $(M, g_{ab})$. A spacetime is called \emph{globally hyperbolic} if it admits a Cauchy surface.

Given a globally hyperbolic spacetime $(M, g_{ab})$. Equation (2) has a well posed initial value formulation, with the initial data $(\phi, \tau)$ on a Cauchy surface $\Sigma$, where $\tau=n^{a}\nabla_{a}\phi$, with $n^{a}$ the unit normal to $\Sigma$. Let the classical phase space $\mathcal{P}$ of the Klein-Gordon theory be the set of pairs $(\phi, \tau)$ of smooth and of compact support functions on $\Sigma$.  Let $\Xi$ be the space of solutions to (2) which arise from initial data in $\mathcal{P}$. $\Xi$ is independent of the choice of $\Sigma$ and each $(\phi, \tau)\in\mathcal{P}$, gives rise to a unique element of $\Xi$ [27]. The symplectic structure $\Omega$ on $\mathcal{P}$ is given by
\begin{gather}
\Omega[(\phi_{1}, \tau_{1}),(\phi_{2}, \tau_{2})]=\int_{\sum}(\tau_{1}\phi_{2}-\tau_{2}\phi_{1})d^{3}x.
\end{gather}
$\Omega$ is conserved for solutions, so (3) gives rise to a bilinear map $\Omega:\Xi\times\Xi\longrightarrow\mathbb{R}$.

We construct a quantum field theory for the Klein-Gordon field as follows:\\
Let $\nu:\Xi\times\Xi\longrightarrow \mathbb{R}$ be a bilinear map such that for all $\psi_{1}\in\Xi$,
\begin{gather}
\nu(\psi_{1}, \psi_{1})=\frac{1}{4}\mathrm{sup}_{\psi_{2}\neq 0}\frac{[\Omega(\psi_{1},\psi_{2})]^{2}}{\nu(\psi_{2}, \psi_{2})}.
\end{gather}
There is a wide class of $\nu$'s satisfying (4) [26]. Given a bilinear map $\nu$ satisfying (4), we complete $\Xi$ in the inner product $2\nu$. By using $\Omega$, we convert the completion of $\Xi$ into a complex Hilbert space $\mathcal{H}$ [27]. Consequently, a projection map $L:\Xi\longrightarrow \mathcal{H}$ gives rise. To define the quantum field theory, we choose the Hilbert space $\mathcal{F}$ to be the symmetric Fock space $\mathcal{F}_{s}(\mathcal{H})=\bigoplus_{n=0}^{+\infty}\,(\mathcal{H}^{\bigotimes_{s}n})$, where $\mathcal{H}^{\bigotimes_{s}n}$ is the subspace of the $n$-fold tensor product $\mathcal{H}^{\bigotimes n}$ consisting of the maps which are totally symmetric in the $n$-variables. Finally, we define the observables $\hat{\Omega}(\psi, .)$ as self-adjoint operators on $\mathcal{F}$ by
\begin{gather*}
\hat{\Omega}(\psi, .)=ia(\overline{L\psi})-ia^{*}(L\psi),
\end{gather*}
where $a, a^{*}$ are the annihilation and creation operators, respectively. However, the construction has the potentially unsatisfactory feature that it appears to depend on our choice of $\nu$. Let $\mathcal{H}_{1}$ and $\mathcal{H}_{2}$ be the Hilbert spaces associated with bilinear maps $\nu_{1}$ and $\nu_{2}$ satisfying (4), respectively. In the case when $\mathrm{dim}\,\mathcal{P}<\infty$, it follows from the Stone-von Neumann theorem [28] that $\{\mathcal{F}_{s}(\mathcal{H}_{1}); \hat{\Omega}_{1}(\psi, .)\}$ and $\{\mathcal{F}_{s}(\mathcal{H}_{2}); \hat{\Omega}_{2}(\psi, .)\}$ are unitarily equivalent, i.e. there is a unitary operator $U$ such that $U^{-1}\hat{\Omega}_{2}(\psi, .)U=\hat{\Omega}_{1}(\psi, .)$ for all $\psi$. But in the case when $\mathrm{dim}\,\mathcal{P}=\infty$ (general quantum field theory in curved spacetime), unitarily inequivalent field theory constructions give rise.

In the usual formulation of a quantum theory, one first constructs states in a Hilbert space $\mathcal{F}$ and then defines observables as operators on $\mathcal{F}$. But in our approach, observables are constructed as elements of an abstract algebra and then states are defined as objects which act upon observables. The key observation that makes our approach applicable is that although unitarily inequivalent field theory constructions exist, the algebraic structure of the field operators in the unitarily inequivalent constructions are the same.

Let $\mathcal{L}(\mathcal{F})$ be the $C^{*}-$algebra of all bounded linear maps on $\mathcal{F}$ with the *-operation corresponding to taking adjoints. Let $\mathcal{W}=\overline{\mathrm{Span}\{\hat{W}(\psi): \psi\in\Xi\}}$, where $\hat{W}(\psi)=e^{i\hat{\Omega}(\psi, .)}$ and the closure is taken in the norm provided by $\mathcal{L}(\mathcal{F})$. $\mathcal{W}$ is a $C^{*}-$subalgebra of $\mathcal{L}(\mathcal{F})$ which is called the \emph{Weyl algebra}. In fact, although two bilinear maps $\nu_{1}$ and $\nu_{2}$ satisfying (4) may define unitarily inequivalent quantum field theory constructions, the associated $C^{*}-$algebras $\mathcal{W}_{1}$ and $\mathcal{W}_{2}$ which give rise are isomorphic [29]. According to this fact, we can define the fundamental observables for quantum field theory in a curved spacetime to be elements of the Weyl algebra $\mathcal{W}$. Then we define a state $w$ of the quantum field to be a linear map $w:\mathcal{W}\longrightarrow \mathbb{C}$ satisfying $w(A^{*}A)\geq 0$ for all $A\in\mathcal{W}$ and $w(I)=1$. A natural question that gives rise is: What is the relationship between the notion of states defined here and states of the Hilbert space $\mathcal{F}=\mathcal{F}_{s}(\mathcal{H})$? We take into account that given any density matrix $\rho$ on $\mathcal{F}$, which carries a representation $\pi:\mathcal{W}\longrightarrow\mathcal{L}(\mathcal{F})$, we obtain a state $w:\mathcal{W}\longrightarrow \mathbb{C}$ by $w(A)=\mathrm{tr}[\rho\pi(A)]$. Therefore, all states arising in all quantum field theory constructions give rise to states defined here. The converse of this result is given by Theorem 2.

The crucial advantage of our approach is that it allows us to treat all states, in particular states arising in unitarily inequivalent quantum field theory constructions, on an equal footing. Thus, one can define the theory without the need to select a preferred construction.
\section{Resolution of bounded linear functionals into normal and singular states}
Let $\pi$ and $\Pi$ be a representation and the universal representation of a $C^{*}-$algebra $\textsf{A}$ on a Hilbert space $H$, respectively, also $\textsf{B}=\Pi(\textsf{A})$ and $\textsf{B}_{\pi}=\pi(\textsf{A})$. Firstly, we need to prove the following two lemmas on extension of the bounded linear mappings.
\begin{lemma}{1}
\emph{There is an isometric isomorphism from $\textsf{B}^{*}$ onto $(\textsf{B}^{-})_{*}$.
}\end{lemma}
\begin{proof}
According to Theorem 2, for every state $\xi$ on $\textsf{A}$, there is a representation $\pi_{\xi}$ with a unit vector $x$ such that $\xi=\omega_{x}\circ \pi_{\xi}$. Hence, $\xi=\omega_{u}\circ \Pi$, where $u=\bigoplus_{\rho\in S(\textsf{A})}u_{\rho}\in\bigoplus_{\rho\in S(\textsf{A})}H_{\rho}$, defined by $u_{\xi}=x, u_{\rho}=0$ for $\rho\neq\xi$. The mapping $\xi\longmapsto \xi\circ \Pi^{-1}$ sends the state space of $\textsf{A}$ onto the state space of $\textsf{B}$. Thereby, every state on $\textsf{B}$ is a vector state. It is well-known that every bounded linear functional on $\textsf{B}$ can be written as a linear combination of at most four states. Every vector state $\omega_{x}$ on $\textsf{B}$ extends to $\omega_{x}$ on $\textsf{B}^{-}$, therefore $f$ extends to the corresponding linear combination of vector states on $\textsf{B}^{-}$. It shows that $f$ extends to a weak-operator continuous linear functional $\widetilde{f}$ on $\textsf{B}^{-}$. This extension is unique by continuity. With $T\in (\textsf{B})_{1}$, we have $|\widetilde{f}(T)|=|f(T)|\leq\|f\|$, then for each $T\in (\textsf{B}^{-})_{1}$, by the Kaplansky density theorem, $|\widetilde{f}(T)|\leq\|f\|$. It results that $\|\widetilde{f}\|\leq \|f\|$ and hence $\|\widetilde{f}\|=\|f\|$. With $g\in(\textsf{B}^{-})_{*}$, let $f=g|_{\textsf{B}}\in \textsf{B}^{*}$. Therefore, $g=\widetilde{f}$ on $\textsf{B}$ and are ultraweakly continuous on $\textsf{B}^{-}$, by Theorem 1, and thus $g=\widetilde{f}$ on $\textsf{B}^{-}$.
\end{proof}
\begin{lemma}{2}
\emph{Let $\textsf{A}$ be a $C^{*}-$subalgebra of $B(H)$, $J$ be a Hilbert space, and $\alpha:\textsf{A}\longrightarrow B(J)$ be an ultraweakly continuous $*-$homomorphism. Then $\alpha$ extends uniquely to an ultraweakly continuous $*-$homomorphism $\widetilde{\alpha}:\textsf{A}^{-}\longrightarrow B(J)$ with $\widetilde{\alpha}(\textsf{A}^{-})=\alpha(\textsf{A})^{-}$.}
\end{lemma}
\begin{proof}
Without loss of generality, we may assume that $\|\alpha\|=1$. For each positive real number $r$, $\alpha$ maps the ball $(\textsf{A})_{r}$ (i.e. ball with radius $r$) into the ball $(B(J))_{r}$. Since $(\textsf{A})_{r}^{-}=(\textsf{A}^{-})_{r}$ and $(B(J))_{r}$ is ultraweakly compact then $\alpha|_{(\textsf{A})_{r}}$ extends uniquely to an ultraweakly continuous mapping $\alpha_{r}:(\textsf{A}^{-})_{r}\longrightarrow(B(J))_{r}$. With $0<r<s$, we have $\alpha_{r}(T)=\alpha_{s}(T)$ for each $T\in (\textsf{A}^{-})_{r}$. Hence, we can define $\widetilde{\alpha}:\textsf{A}^{-}\longrightarrow B(J)$ by $\widetilde{\alpha}(T)=\alpha_{r}(T)$ where $T\in (\textsf{A}^{-})_{r}$. Given $S, T\in\textsf{A}^{-}$, there are nets $\{S_{i}\}, \{T_{j}\}$ in $\textsf{A}$, converging ultraweakly to $S, T$, respectively, such that $\|S_{i}\|\leq \|S\|, \|T_{j}\|\leq \|T\|$. Given scalars $s, t$, let $r=\mathrm{max}\{\|T\|, \|S\|, |s|\|S\|+|t|\|T\|\}$, therefore, $S_{i}, T_{j}, sS_{i}+tT_{j}$ all lie in $(\textsf{A}^{-})_{r}$, and $\widetilde{\alpha}(sS_{i}+tT_{j})=s\widetilde{\alpha}(S_{i})+t\widetilde{\alpha}(T_{j})$, by the linearity of $\alpha=\widetilde{\alpha}|_{\textsf{A}}$. When $S_{i}\rightarrow S$ and then $T_{j}\rightarrow T$, we have $\widetilde{\alpha}(sS+tT)=s\widetilde{\alpha}(S)+t\widetilde{\alpha}(T)$, by the ultraweak continuity of $\widetilde{\alpha}|_{(\textsf{A}^{-})_{r}}$. Hence, $\widetilde{\alpha}$ is a linear mapping. Similarly, we can prove that
\begin{align*}
\widetilde{\alpha}(T^{*})=\widetilde{\alpha}(T)^{*},\quad \widetilde{\alpha}(ST)=\widetilde{\alpha}(S)\widetilde{\alpha}(T)\quad S, T\in \textsf{A}^{-},
\end{align*}
by the ultraweak continuity of $\widetilde{\alpha}$ and the adjoint operation, and the separate ultraweak continuity of multiplication. Hence, $\widetilde{\alpha}$ is a $*-$homomorphism. It is clear that $\widetilde{\alpha}$ extends $\alpha$ and is norm decreasing, then $\|\widetilde{\alpha}\|=\|\alpha\|=1$. If $f$ is an ultraweakly continuous linear functional on $B(J)$, the composition $f \circ \widetilde{\alpha}$ is ultraweakly continuous on $(\textsf{A}^{-})_{r}$ for each $r>0$. Then by Theorem 1, $f \circ \widetilde{\alpha}$ is ultraweakly
continuous on $\textsf{A}^{-}$. It results that $\widetilde{\alpha}$ is ultraweakly continuous on $\textsf{A}^{-}$, and then $\alpha=\widetilde{\alpha}|_{\textsf{A}}$ is ultraweakly continuous on $\textsf{A}$.\par
With $\textsf{M}=\alpha(\textsf{A})$, we have $(\textsf{M})_{1}=\alpha((\textsf{A})_{1})$. Since $\widetilde{\alpha}$ is ultraweakly continuous and $(\textsf{A}^{-})_{1}$ is ultraweakly compact,  $\widetilde{\alpha}((\textsf{A}^{-})_{1})$ is an ultraweakly compact set containing $(\textsf{M})_{1}$ and then containing its closure $(\textsf{M}^{-})_{1}$. Thus $\textsf{M}^{-}\subseteq \widetilde{\alpha}(\textsf{A}^{-})$ and therefore $\widetilde{\alpha}(\textsf{A}^{-})=\textsf{M}^{-}=\alpha(\textsf{A})^{-}$, by the ultraweak continuity of $\widetilde{\alpha}$.
\end{proof}\par
In the proof of the following theorem, we apply the fact that the weak-operator closed left ideals in a von Neumann algebra $R$ are left principal ideals of the form $RP$ for a projection $P$ in $R$.
\begin{theorem}{3}
\emph{There is a projection $E$ in the center of $\textsf{B}^{-}$, and a $*-$isomorphism $\varphi:\textsf{B}^{-}E\longrightarrow\textsf{B}_{\pi}^{-}$, such that $\varphi(\Pi(T)E)=\pi(T)$ for each $T$ in $\textsf{A}$.}
\end{theorem}
\begin{proof}
Since $\Pi$ is a faithful representation, $\mu=\pi\circ \Pi^{-1}:\textsf{B}\longrightarrow\textsf{B}_{\pi}$ is a $*-$homomorphism and then is bounded. If $f$ is an ultraweakly continuous linear functional on $\textsf{B}_{\pi}$, the linear functional $f\circ \mu$ on $\textsf{B}$ is bounded and thus, by applying the argument of Lemma 1, is ultraweakly continuous. It results that $\mu$ is ultraweakly continuous (see Remark 1). Therefore, $\mu$ extends to an ultraweakly continuous $*-$homomorphism $\widetilde{\mu}:\textsf{B}^{-}\longrightarrow\textsf{B}_{\pi}^{-}$, by Lemma 2. The kernel of $\widetilde{\mu}$ has the form $\textsf{B}^{-}P$ for some projection in the center of $\textsf{B}^{-}$. Let $E=I-P$, then for each $T$ in $\textsf{B}^{-}$,
\begin{align}
\widetilde{\mu}(T)=\widetilde{\mu}(TI)=\widetilde{\mu}(TE+TP)=\widetilde{\mu}(TE),
\end{align}
We define $\varphi=\widetilde{\mu}|_{\textsf{B}^{-}E}$ that according to (5), has the same range as $\widetilde{\mu}$. Since $\textsf{B}^{-}E\cap\textsf{B}^{-}P=\{0\}$, therefore $\varphi$ is one-to-one, and thus is a $*-$isomorphism from $\textsf{B}^{-}E$ onto $\textsf{B}_{\pi}^{-}$. For each $T$ in $\textsf{A}$,
\begin{align*}
\varphi(\Pi(T)E)=\widetilde{\mu}(\Pi(T)E)=\widetilde{\mu}(\Pi(T)E+\Pi(T)P)=\widetilde{\mu}(\Pi(T))=\mu(\Pi(T))=\pi(T).
\end{align*}
\end{proof}
\begin{remark}{1}
Given $x, y\in H$, let $f(T)=\langle T(x), y\rangle$, where $T\in\textsf{B}_{\pi}^{-}$. By extending $f\circ \mu$ to an ultraweakly continuous linear functional $g$ on $\textsf{B}^{-}$, or equivalently weak-operator continuous on $(\textsf{B}^{-})_{1}$, it results that $\mu$ is ultraweakly continuous.
\end{remark}\par
Let $f$ be a bounded linear functional on $\textsf{A}$. Then $g=f\circ\Pi^{-1}$ is an ultraweakly continuous linear functional on $\textsf{B}$. According to Lemma 1, we can extend $g$ to $\widetilde{g}$ such that
\begin{align}
\widetilde{g}(\Pi(T))=f(T),\quad T\in\textsf{A}.
\end{align}
Let $E$ be the projection in Theorem 3 when $\pi$ is the inclusion mapping and define $\lambda_{E}$ on $\textsf{A}^{*}$ in the following manner, $\lambda_{E}(f)(T):=\widetilde{g}(\Pi(T)E)$, $T\in\textsf{A}$.
\begin{theorem}{4}
\emph{Let $E$ be the projection in the center of $\textsf{B}^{-}$ when $\pi:\textsf{A}\longrightarrow B(H)$ is the inclusion mapping in Theorem 3. A bounded linear functional $f$ on $\textsf{A}$  is ultraweakly continuous if and only if $\lambda_{E}(f)=f$.}
\end{theorem}
\begin{proof}
Suppose that $f$ is ultraweakly continuous. Then $f$ extends to an ultraweakly continuous linear functional $\overline{f}$ on $\textsf{A}^{-}$ . According to Theorem 3, we have
\begin{align}
T=\pi(T)=\varphi(\Pi(T)E),\quad T\in\textsf{A}.
\end{align}
It follows from (6), (7) that
\begin{align*}
\widetilde{g}(\Pi(T))=f(T)=\overline{f}(T)=\overline{f}(\varphi(\Pi(T)E)),\quad T\in \textsf{A}.
\end{align*}
Accordingly,
\begin{align}
\widetilde{g}(S)=\overline{f}(\varphi(SE)),\quad S\in \textsf{B}.
\end{align}
Because of the ultraweak continuity of $\widetilde{g}, \overline{f},$ and $\varphi$, the equality (8) holds for each $S\in\textsf{B}^{-}$. When $S$ is replaced by $SE$, the right-hand side of (8) is unchanged, therefore
\begin{align*}
\widetilde{g}(SE)=\widetilde{g}(S),\quad S\in \textsf{B}^{-}.
\end{align*}
Particularly,
\begin{align*}
f(T)=\widetilde{g}(\Pi(T))=\widetilde{g}(\Pi(T)E)=\lambda_{E}(f)(T),\quad T\in \textsf{A}.
\end{align*}
Now, suppose that $\lambda_{E}(f)=f$ holds. It follows from (7) that $f=\widetilde{g}\circ \varphi^{-1}|_{\textsf{A}}$ and the ultraweak continuity of $\widetilde{g}$ and $\varphi^{-1}$ results that $f$ is ultraweakly continuous.
\end{proof}\par
The mapping $f\longmapsto \lambda_{E}(f)$ is a norm-decreasing projection from $\textsf{A}^{*}$ onto a closed subspace $\textsf{A}^{N}$ consisting of ultraweakly continuous linear functionals or, by Theorem 1, normal states on $\textsf{A}$. Non-zero elements of the complementary closed subspace
\begin{align*}
\textsf{A}^{S}=\{f\in \textsf{A}^{*}: \lambda_{E}(f)=0\},
\end{align*}
are described as \emph{singular} states on $\textsf{A}$. In fact, this projection maps the dual space onto the space of normal states parallel to the space of singular states.
\begin{corollary}{1}([5])
\emph{Each $f$ in $\textsf{A}^{*}$ can be decomposed, uniquely, in the form $f=f_{N}+f_{S}$, with $f_{N}$ in $\textsf{A}^{N}$ and $f_{S}$ in $\textsf{A}^{S}$. }
\end{corollary}
\begin{proof}
Let
\begin{align*}
f_{N}=\lambda_{E}(f),\quad f_{S}=(I-\lambda_{E})(f).
\end{align*}
\end{proof}\par
As a result of Corollary 1 , every pure state is a normal or a singular state. According to Theorem 1, singular states annihilate all one-dimensional projections, and therefore all compact operators. Let $H$ be an infinite dimensional separable Hilbert space, $T \in B(H)$ a positive compact operator, and $\{\eta (k, T)\}_{k=1}^{\infty}$ an eigenvalue sequence of $T$ counted with multiplicities and arranged in a non-increasing order. Suppose that
\begin{align*}
\mathrm{sup}_{n\geq 1}\;\frac{1}{\mathrm{log}(1+n)}\sum_{k=1}^{n}\eta(k, T)<\infty.
\end{align*}
If $\rho$ is a state on the space $l_{\infty}$ of bounded sequences such that the following conditions hold,\\
(i) $\rho$ vanishes on finitely supported sequences, i.e. $S=\{(x_{1}, x_{2},...) : \exists N\; \mathrm{s. t.}\; x_{n}=0\; \forall n\geq N\}$,\\
(ii) $\rho(x_{1}, x_{2},...)=\rho(x_{1}, x_{1}, x_{2}, x_{2},...)$ for each $x\in l_{\infty}$, then
\begin{align*}
\mathrm{Tr}_{\rho}(T):=\rho\left(\frac{1}{\mathrm{log}(1+n)}\sum_{k=1}^{n}\eta(k, T)\right),
\end{align*}
is a non-trivial state that vanishes all compact operators and thereby is a singular state [13].
\subsection{Physical interpretation and example}
We can think of a non-relativistic particle localized at a sharp point, as witnessed by the expectations of all continuous functions of position. Extending from this algebra to all bounded operators, we get a singular state with sharp position, but “inﬁnite momentum”, i.e. the probability assigned to ﬁnding the momentum in any given ﬁnite interval is zero [14]. This shows that the probability measure on the momentum space induced by such a state is only ﬁnitely additive, but not $\sigma-$additive. This is typical for singular states.
\subsubsection{Arising in quantum dynamics}
Consider the evolution of the quantum many-body system of Bose particles with the following Hamiltonian:
\begin{gather*}
\mathbb{H}=\int\psi^{*}(x)h\psi(x)dx+\frac{1}{2}\int\upsilon(x-y)\psi^{*}(x)\psi^{*}(y)\psi(x)\psi(y)dxdy,
\end{gather*}
where (i) $\psi$ and $\psi^{*}$ are the annihilation and creation operators, respectively, satisfying the canonical commutation relations (CCR, see Section 6.1), as follows:
\begin{gather*}
[\psi(x), \psi^{*}(y)]=\delta(x-y),\quad [\psi(x), \psi(y)]=[\psi^{*}(x), \psi^{*}(y)]=0.
\end{gather*}
(ii) $h:=-\triangle+V(x)$ acting on the variable $x$, defined on the symmetric (or Bosonic) Fock space $\mathcal{F}:=\bigoplus_{n=0}^{+\infty}(L^{2}(\mathbb{R}^{d}, \mathbb{C}))^{\bigotimes_{s}n}$, where $-\triangle$ and $V(x)$ are the Laplacian and the external potential, respectively. Under some conditions on $\upsilon$ and $V$, it follows that $\mathbb{H}$ is a self-adjoint operator [30].

Here states, as positive linear functionals $w$ on the Weyl CCR algebra, correspond either to nonzero densities and temperatures, or to a finite number of particles as in the case of Bose-Einstein condensate experiments [31], in traps. In the latter case, states are represented by density operators, i.e. there is a density operator $D$ on $\mathcal{F}$, such that for all observables $A$ we have
\begin{gather*}
w(A)=\mathrm{tr}(DA).
\end{gather*}
As an example for singular states, we can refer to the equilibrium state of a free Bose gas in infinite space at finite density and temperature. This state is singular (with respect to Fock space) because the probability for finding only a finite number of particles in it, is zero.
\section{An application in quantum information theory}
The no-cloning theorem says that there is no unitary operator that makes a perfect copy of an unknown pure state [15]. However perfect quantum cloning is impossible, it is possible to perform imperfect cloning. Accordingly, several studies on approximate duplication of quantum states have been done, e.g. see these impressive works [16, 17, 18].
\subsection{Phase space and Weyl operators}
A \emph{phase space} consists of a real vector space $V$ of dimension $2n$ ($n$ is the number of modes) and an antisymmetric bilinear form $\beta:V\times V \longrightarrow \mathbb{R}$ that defines the scalar product $\beta(x, y)=\sum_{k, l=1}^{2n} x_{k}^t \, \beta_{k, l} \, y_{l}$ ($^{t}$ is transposition) where $\beta_{k, l}=\beta(e_{k}, e_{l})$ and $\{e_{k}\}$ is an orthonormal basis for $V$.

A quantum system is described by the \emph{canonical commutation relations} (CCR) between the fields operators $R_{k}$ ($k=1, 2, ..., 2n$) as $[R_{k}, R_{l}]=i \beta_{k, l} I$ where $I$ is the identity operator. In quantum mechanics, fields operators are written in the vector form as $\overrightarrow{R}=(Q_{1}, P_{1}, Q_{2}, P_{2}, ..., Q_{n}, P_{n})$ or $\overrightarrow{R}=(Q_{1},Q_{2}, ..., Q_{n}, P_{1}, P_{2}, ..., P_{n})$ where $Q_{i}, P_{i}$ are position and momentum operators of each mode, respectively $(i=1, 2, ..., n)$.

\emph{Weyl operators} as a family of bounded and unitary operators are defined in the following manner $W_{x}=e^{ix^{t} \textbf{B} \overrightarrow{R}}$ where $x\in V$, $\textbf{B}=[\beta(e_{k}, e_{l})]$ and $W_{0}=I$ [17]. For $x=(q_{1}, p_{1}, q_{2}, p_{2}, ..., q_{n}, p_{n})$ the standard form of the Weyl operators in quantum mechanics is written as $W_{x}=\mathrm{exp}(i\sum_{k=1}^{n}(q_{k}P_{k}-p_{k}Q_{k}))$. It follows from the CCR that the Weyl operators satisfy the following relations:
\begin{align*}
W_{x}W_{y}=e^{\frac{-i\beta(x, y)}{2}}W_{x+y},\quad W_{x}W_{y}=e^{-i\beta(x, y)}W_{y}W_{x},\quad W_{x}^{*}=W_{-x}.
\end{align*}

In our study, the Weyl operators have the \emph{irreducibility property}, i.e. if for every $x$ in $V$, $[W_{x}, A]=0$ then $A\propto I$.
\subsection{Covariant cloning}
Let $(\mathbb{R}^{2}, \beta_{{\mathrm{in}}})$ be the phase space of the input system and $(\mathbb{R}^{2n}, \beta_{\mathrm{out}})$ be the phase space of the output system including $n$ subsystems ($n$ modes) where
\begin{align*}
\beta_{\mathrm{out}}(x, y)=\sum_{j=1}^{n}\beta_{{\mathrm{in}}}(x_{j}, y_{j}),\quad x=(x_{1}, x_{2},...,x_{n}), \quad y=(y_{1}, y_{2}, ..., y_{n}).
\end{align*}

Let $A_{{\mathrm{in}}}$ and $A_{\mathrm{out}}$ be the CCR algebras of observables on the input and output systems, respectively. Accordingly, the state spaces are $S(A_{{\mathrm{in}}})=B^{*}(H_{\mathrm{in}}), S(A_{\mathrm{out}})=B^{*}(H_{\mathrm{out}})$ where $H_{\mathrm{in}}=L^2(\mathbb{R}^{2})$ and $H_{\mathrm{out}}=H_{{\mathrm{in}}}^{\bigotimes n}\simeq L^2(\mathbb{R}^{2n})$.

A \emph{1 to $n$ cloner} is a quantum channel, i.e. completely positive and trace-preserving map, that in the Heisenberg picture, maps $A_{\mathrm{out}}$ onto $A_{{\mathrm{in}}}$, $\Phi: A_{\mathrm{out}}\longrightarrow A_{{\mathrm{in}}}$ and in the Schr\"{o}dinger picture, maps $S(A_{{\mathrm{in}}})$ onto $S(A_{\mathrm{out}})$, $\Phi_{*}: S(A_{{\mathrm{in}}})\longrightarrow S(A_{\mathrm{out}})$.

To compare the input and output states of the cloner we need a functional that determines how well they coincide. For two states $\rho_{1}, \rho_{2}$, this functional is defined as
\begin{gather}
F(\rho_{1}, \rho_{2})=\left(\mathrm{tr}\left[(\rho_{1}^{\frac{1}{2}}\rho_{2}\rho_{1}^{\frac{1}{2}})^{\frac{1}{2}}\right]\right)^{2},\notag
\end{gather}
and is called the \emph{fidelity} [19]. In our study, at least one of the states is pure and this definition is reduced to $F(\rho_{1}, \rho_{2})=\mathrm{tr}[\rho_{1}\rho_{2}]$. Accordingly, for an input coherent state $\rho=|\alpha\rangle\langle\alpha|$, the (joint) fidelity between the output state of the cloner, $\Phi_{*}(\rho)$ and $\rho^{\bigotimes n}$ is defined as the expectation value of the operator $\rho^{\bigotimes n}$ in the output state of the cloner as follows  $F(\Phi, \rho)=\mathrm{tr}\left[\Phi_{*}(\rho)\rho^{\bigotimes n}\right]$. Let $C=\{|\alpha\rangle\langle \alpha|: \alpha\in \mathbb{R}^{2}\}$, we choose $F(\Phi)=\mathrm{inf}_{\rho\in C}\,F(\Phi, \rho)$ as the worst-case overlap between the input and output states.

For $x \in \mathbb{R}^{2}$, we define the \emph{shift cloner} $\Phi_{*}^{x}$ by
\begin{align*}
\Phi_{*}^{x}(\rho)={W_{x}^{\bigotimes n}}^{*}\Phi_{*}(W_{x} \rho W_{x}^{*})W_{x}^{\bigotimes n}.
\end{align*}
A cloner $\Phi_{*}$ is called \emph{covariant} if $\Phi_{*}^{x}(\rho)=\Phi_{*}(\rho)$ for every $x, \rho$.
\begin{remark}{2}
For a covariant cloner $\Phi_{*}$ and given observable $O$ and state $\rho$, we have
\begin{align*}
\mathrm{tr}\left[\rho\,\Phi(O)\right]=\mathrm{tr}\left[\Phi_{*}(\rho)\,O\right]=\mathrm{tr}\left[{W_{x}^{\bigotimes n}}^{*}\Phi_{*}(W_{x} \rho W_{x}^{*})W_{x}^{\bigotimes n}O\right]=\mathrm{tr}\left[\rho W_{x}^{*} \Phi(W_{x}^{\bigotimes n}O{W_{x}^{\bigotimes n}}^{*})W_{x}\right],
\end{align*}
hence, it results that $\Phi(O)=W_{x}^{*} \Phi(W_{x}^{\bigotimes n}O{W_{x}^{\bigotimes n}}^{*})W_{x}$ and it means that $\Phi$ is also covariant.
\end{remark}

Let $\Phi_{*}$ be a general cloner in the form $\Phi_{*}(\rho)=(\Phi_{*}^{\mathcal{N}}\bigotimes \Phi_{*}^{\mathcal{S}})(\rho)$, where $\Phi_{*}^{\mathcal{N}}$ is a normal cloner, i.e. the respective output state is a positive trace-class operator, and $\Phi_{*}^{\mathcal{S}}$ is a singular cloner, i.e. the respective output state is singular. Consequently, a natural question gives rise: What is the structure of the optimal covariant cloner? This question will be answered below in case $n=2$, for 1 to 2 covariant cloners.

We define $\mathcal{K}\subset B(H_{\mathrm{out}})$ as the subalgebra generated by the identity operator and operators of the form $I\bigotimes K$ and $K \bigotimes I$, where $K$ is a compact operator and $I$ is the identity operator. A general element $\mathcal{A}$ of $\mathcal{K}$ is represented by the notation $\mathcal{A}=\sum_{M\subseteq \{1, 2\}}\,K_{M}\bigotimes I^{\bigotimes M^{c}}$ in which $K_{M}$ is the compact component that denotes tensor product of compact operators in the tensor factors indexed by $M$ and $I^{\bigotimes M^{c}}$ denotes tensor product of the rest. For a state $\omega$ on the output states, $\omega$ can be decomposed into parts which are labeled by a set of tensor factors $N$ on which the state is normal and hence described by a positive trace-class operator $\Lambda_{N}$ and on $N^{c}$ is singular, as follows:
\begin{align}
\omega(\mathcal{A})=\omega\left(\sum_{M\subseteq \{1, 2\}}\,K_{M}\bigotimes I^{\bigotimes M^{c}}\right)=\sum_{N\subseteq \{1, 2\}}\sum_{M\subseteq N}\mathrm{tr}\left[\Lambda_{N}\left(K_{M}\bigotimes I^{\bigotimes (N-M)}\right)\right],
\end{align}
we take into account that the expectation value of compact operators in a singular output state is zero, hence in the above sum, we consider only subsets $M$ of $N$. Since $\omega$ is a state then the normalization condition gives rise to
\begin{align}
\sum_{N\subseteq\{1, 2\}} \mathrm{tr}[\Lambda_{N}]=1.
\end{align}

Let $\omega=\Phi_{*}(\rho)|_{\mathcal{K}}$ and define $\Phi_{*}^{N}:\rho\mapsto \Lambda_{N}$ that maps $\rho$ to the unique positive trace-class operator $\Lambda_{N}$ from (9). Since $\Phi_{*}$ is covariant, then $\Phi_{*}^{N}$ is so. Although $\Phi_{*}(\rho)$ is normalized, but $\Phi_{*}^{N}(\rho)$ is not normalized. Let $\Phi^{N}(I)=X_{N}$, then it follows from Remark 2 and (10) that for every $\rho$,
\begin{align*}
\mathrm{tr}[\rho X_{N}]=\mathrm{tr}[\Phi_{*}^{N}(\rho)]\leq 1,
\end{align*}
and since $\Phi_{*}^{N}$ is covariant, then $X_{N}$ commutes with all Weyl operators. Hence, it follows from the irreducibility property that $X_{N}=c_{N}I$ for some constant $0<c_{N}\leq 1$, which does not depend on the input state $\rho$. We define $\Psi_{*}^{N}(\rho)=\frac{\Phi_{*}^{N}(\rho)}{c_{N}}=\frac{\Lambda_{N}}{c_{N}}$, which is a normal cloner (since $\frac{\Lambda_{N}}{c_{N}}$ is a density operator).

Now we calculate the fidelity: Since $\Phi_{*}$ is covariant, by applying Remark 2 and the equality $W_{\alpha}|0\rangle=|\alpha\rangle$,
\begin{align*}
F(\Phi)=\mathrm{inf}_{\rho\in C}\,F(\Phi,\rho)=F(\Phi,|0\rangle\langle 0|)=\Phi_{*}(|0\rangle\langle 0|)(|\alpha\rangle\langle\alpha|^{\bigotimes 2}).
\end{align*}
For $N=\{1, 2\}$ by applying (9) and compactness of $|\alpha\rangle\langle\alpha|^{\bigotimes 2}$,
\begin{align*}
F(\Phi)=\mathrm{tr}[\Lambda_{N}\,|\alpha\rangle\langle\alpha|^{\bigotimes 2}]=c_{N}\mathrm{tr}\left[\Psi_{*}^{N}(|0\rangle\langle 0|)|\alpha\rangle\langle\alpha|^{\bigotimes 2}\right]=c_{N}F(\Psi^{N}).
\end{align*}
To achieve the optimal cloner, we need to increase the fidelity, hence $c_{N}=1$ and by (10), for $N\varsubsetneq \{1, 2\}$, $c_{N}=0$. It shows that $\Phi_{*}=\Psi_{*}^{N}$, where $N=\{1, 2\}$ and it means that the optimal cloner can not have a singular component and it is purely normal.

\section{Discussion}
Let $H$ be a separable infinite dimensional Hilbert space. We say that a pure state $\rho$ on $B(H)$ is \emph{diagonalizable}, if $\rho(T)=\mathrm{lim}_{\mathcal{U}}\langle T(e_{n}), e_{n}\rangle$ for some orthonormal basis $\{e_{n}\}$ of $H$ and some ultrafilter $\mathcal{U}$ in $\mathbb{N}$. An \emph{atomic masa} (maximal abelian self-adjoint subalgebra) is the set of all operators on $H$ that are diagonalized with respect to some orthonormal basis.
\begin{theorem}{5}([9, 10])
For every pure state $\rho$ on $B(H)$, there is a basis $\{e_{n}\}$ of $H$ and measure $\mu$ such that
\begin{gather}
\rho(T)=\int_{\mathbb{N}}\langle T(e_{n}), e_{n}\rangle d\mu.\notag
\end{gather}
\end{theorem}

By applying the following notation $\rho(T)=\langle \rho, T\rangle$, and $\rho_{n}(T)=\langle T(e_{n}), e_{n}\rangle$, the result of Theorem 5, gives rise to:
\begin{align*}
\langle \rho, T\rangle=\int_{\mathbb{N}}\langle \rho_{n}, T\rangle d\mu.
\end{align*}
It is called the representation of $\rho$ \emph{in the sense of the Gelfand-Pettis integral}.

According to the \textbf{Kadison-Singer conjecture} which was proved in 2013 [20] every pure state on an atomic masa of $B(H)$ has a unique extension to a pure state on $B(H)$. Nevertheless, assuming the continuum hypothesis, there is a pure state on $B(H)$ whose restriction to any masa is not pure [21]. It follows from this discussion that:\\
There is a (singular) pure state on $B(H)$ which can not be represented in the sense of the Gelfand-Pettis integral by an orthonormal basis.
\subsection{Formulation of quantum mechanics in terms of worlds (continue)}
From the Kadison-Singer conjecture it follows that in addition to vector states, singular states also have unique extensions to $B(H)$. It results that unlike observables which are considered relative to a certain world, states can be (uniquely) extended to the whole system and so of \emph{absolute} meaning. Although singular states do not appear in the conventional formalism of quantum mechanics, but they are defined, based on the notion of world by the Stone-\v{C}ech compactification and are (mathematically) constructed by the notion of Banach limit [22]. Thus, realization of such states will verify the necessity of our conceptual approach to worlds of a quantum system.
\subsubsection{Measurement}
Measurement, i.e. making observation, in a certain world is done by applying a quantum state as an external observation. Let $w$ be a pure state in a world $W=\{e_{j}\}$. Every observable $O$ relative to $W$ takes a definite value $w(O)$, particularly $w_{j}(O)$ is the eigenvalue of $O$ corresponding to eigenvector $e_{j}$ ($w_{j}$ is the vector state corresponding to $e_{j}$). For an observable $O^\prime$ relative to another world $W^{\prime}=\{e_{k}^{\prime}\}$, since $w$ is a state in $W$, the information represented by $w$ is incomplete for $W^\prime$. By the Kadison-Singer conjecture, $w$ has a unique extension to $B(H)$ which is also denoted by $w$. Therefore, the expectation value of the new measurement is $\langle O^{\prime}\rangle_{w}:=w(O^\prime)$. Particularly, if $w=w_{j}$ then
\begin{gather*}
w_{j}(O^{\prime})=\langle e_{j}|O^\prime|e_{j}\rangle=\langle e_{j}|\sum_{k}\lambda_{k}|e_{k}^\prime\rangle\langle e_{k}^\prime|e_{j}\rangle=\sum_{k}\lambda_{k}|\langle e_{j}|e_{k}^\prime\rangle|^{2}.
\end{gather*}
Hence, the probability that $O^\prime$ takes the value $\lambda_{k}$ is $|\langle e_{j}|e_{k}^\prime\rangle|^{2}$. It is the same as the usual Born rule.

A natural question gives rise: What occurs in the measurement process? Let $w$ be a state in a world $W=\{e_{j}\}$ and the observer wants to observe an observable $O^\prime$ relative to another world $W^\prime=\{e_{k}^\prime\}$. We take into account that $O^\prime$ relative to $W^\prime$ has a definite value, then we can assign a state $w^\prime$ representing the information obtained by the observer in the world $W^\prime$. Similarly, $w^\prime$ has a unique extension to $B(H)$ which is also denoted by $w^\prime$. The next measurement will only depend on $w^\prime$ and it does not depend on $w$.

It follows that in the measurement process we have only change of information and in contrary of conventional quantum mechanics no collapse of wave packets occurs [23].
\section{Conclusion}
It is shown the resolution of bounded linear functionals into normal and singular states. In fact, there is a projection in the von Neumann algebra generated by applying the GNS construction which maps the dual space onto the space of normal states parallel to the space of singular states. Singular states are mathematically constructed and like vector states, have unique extensions to the whole system. Realization of singular states verifies the necessity of approaching to worlds of a quantum system. Quantum field theory in curved spacetime is constructed based on a Hilbert space $\mathcal{F}$ which is associated with a bilinear map $\nu$. Our treatment allows us to formulate quantum field theory in a manner which does not require the specification of even a preferred unitary equivalence class of bilinear maps. Consequently, one can define quantum theory without the need to select a preferred construction. States correspond to a finite number of particles as in the case of BEC (Bose-Einstein condensate) experiments, in traps are normal states. The equilibrium state of a free Bose gas in infinite space at finite density and temperature is a singular state. As an application in quantum information theory it results that, in the class of general covariant cloners, the optimal cloner can not have a singular component and it is purely normal. There is a singular state which can not be represented in the sense of the Gelfand-Pettis integral by an orthonormal basis. In the formalism of quantum worlds, the process of measurement only changes the information and no collapse of wave packets occurs.

\begin{acknowledgement}{}
I would like to thank Professors Grigori G. Amosov, Vsevolod Zh. Sakbaev, and Oleg G. Smolyanov for motivating to write this paper.
\end{acknowledgement}

\end{justify}


\end{document}